\documentclass[runningheads]{llncs}
\pdfoutput=1
\usepackage{times}
\usepackage{amsmath,amssymb,txfonts}
\usepackage{xytree}
\usepackage{listings}
\usepackage{algpseudocode}
\usepackage{graphicx}
\usepackage{xspace}
\usepackage{tabularx}
\usepackage{paralist}
\usepackage{multirow}
\usepackage{calc}
\usepackage{bbm}
\usepackage{cite}

\lstdefinelanguage{scala}{
  alsoletter={@,=,>,:,-},
  morekeywords={let,abstract, case, catch, class, def, do, else, extends, final, finally, for, if, implicit, import, match, new, null, object, 
override, package, private, protected, requires, return, sealed, super, this, throw, trait, try, true, type, val, var, while, yield, domain, Boolean,
postcondition, precondition,invariant, constraint, assert, forAll, in, _, return, @generator, ensure, require, ensuring, given, have, =>, continue, returns, havoc,:-},
  sensitive=true,
  morecomment=[l]{//},
  morecomment=[s]{/*}{*/},
  morestring=[b]"
}

\newcommand{\codestyle}{\small\sffamily}

\renewcommand{\emph}[1]{\textit {#1}}

\newcommand{\Z}{{\mathbb Z}}

\newif\ifLongVersion\LongVersionfalse

\newcommand{\constraint}{constraint}
\newcommand{\constraintLang}{\constraint\ language}
\newcommand{\constraintLangs}{\constraintLang s}
\newcommand{\Constraint}{Constraint}
\newcommand{\ConstraintLang}{\Constraint\ language}
\newcommand{\ConstraintLangs}{\ConstraintLang s}
\newcommand{\Con}{\ensuremath{\mathit{Constr}}} 
\newcommand{\constraints}{\constraint s}
\newcommand{\relsym}{relation symbol}
\newcommand{\relsyms}{\relsym s}

\newcommand{\unicl}{\ensuremath{\mathit{Cl}_\forall}}

\newcommand{\fv}[1]{\ensuremath{\mathit{fv}(#1)}}
\newcommand{\positions}[1]{\ensuremath{\mathit{positions}(#1)}}
\newcommand{\under}[2]{\ensuremath{#1 \!\downarrow\! #2}}

\newcommand{\expl}{\leftarrow}
\newcommand{\pand}{ \land}
\newcommand{\por}{ \lor}

\newcommand{\ClauseSet}{{\cal HC}}

\algtext*{EndWhile}
\algtext*{EndIf}
\algtext*{EndFor}
\algtext*{EndProcedure}
\algtext*{EndFunction}

\title{Disjunctive Interpolants for Horn-Clause Verification\\
 (Extended Technical Report)}

\author{Philipp R\"ummer\inst{1} 
\and Hossein Hojjat\inst{2} \and 
  Viktor Kuncak\inst{2}}
\authorrunning{R\"ummer, Hojjat, Kuncak}
\institute{Uppsala University, Sweden \and
Swiss Federal Institute of Technology Lausanne (EPFL)}

\begin{document}
\maketitle

\begin{abstract}
One of the main challenges in software verification is
efficient and precise compositional analysis of programs
with procedures and loops. Interpolation methods remains one
of the most promising techniques for such verification, and
are closely related to solving Horn clause constraints.  We
introduce a new notion of interpolation, disjunctive
interpolation, which solve a more general class of problems
in one step compared to previous notions of interpolants,
such as tree interpolants or inductive sequences of interpolants.
We present algorithms
and complexity for construction of disjunctive interpolants,
as well as their use within an abstraction-refinement loop.
We have implemented Horn clause verification algorithms that
use disjunctive interpolants and evaluate them on
benchmarks expressed as Horn clauses over the theory
of integer linear arithmetic. 
\end{abstract}

\section{Introduction}

Software model checking has greatly benefited from the
combination of a number of seminal ideas: automated abstraction
through theorem proving \cite{GrafSaidi97ConstructionAbstractStateGraphsPVS}, 
exploration of finite-state
abstractions, and counterexample-driven refinement 
\cite{BallETAL02RelativeCompletenessAbstractionRefinement}. Even
though these techniques can be viewed independently, the
effectiveness of verification  has been
consistently improving by providing more sophisticated
communication between these steps. Often,
carefully chosen search aspects are being pushed
into a learning-enabled constraint solver, resulting in better
overall verification performance. An essential
advance was to use interpolants derived from
unsatisfiability proofs to refine the abstraction
\cite{Henzinger:2004}.  In recent years,
we have seen significant progress in interpolating methods for
different logical constraints \cite{DBLP:journals/tocl/CimattiGS10,bonacina12,iprincess2011,tree-interpolants}, and a wealth of more
general forms of  interpolation
\cite{tree-interpolants,DBLP:conf/popl/HeizmannHP10,DBLP:conf/sas/AlbarghouthiGC12}. 
In this paper
we identify a new notion, \emph{disjunctive interpolants}, which are
more general than tree interpolants and inductive sequences of
interpolants.  Like tree
interpolation~\cite{tree-interpolants,DBLP:conf/popl/HeizmannHP10}, a
disjunctive interpolation query is a tree-shaped constraint specifying the
interpolants to be derived; however, in disjunctive
interpolation, branching in the tree can represent both conjunctions and
disjunctions.  We present an algorithm for solving the interpolation
problem, relating it to a subclass of recursion-free Horn clauses.
We then consider solving general recursion-free Horn clauses and show that this problem
is solvable whenever the logic admits interpolation. We establish
tight complexity bounds for solving recursion-free Horn clauses for propositional
logic (PSPACE) and for integer linear arithmetic (co-NEXPTIME). In contrast,
the disjunctive interpolation problem remains in coNP for these logics.
We also show how to use solvers for recursion-free Horn clauses to
verify recursive Horn clauses using counterexample-driven predicate abstraction. We
present an algorithm and experimental results on publicly available benchmarks.

\subsection{Related Work}

There is a long line of research on Craig \textbf{interpolation}
methods, and generalised forms of interpolation, tailored to
verification. For an overview of interpolation in the presence of
theories, we refer the reader to
\cite{DBLP:journals/tocl/CimattiGS10,iprincess2011}. Binary Craig
interpolation for implications~$A \to C$ goes back to
\cite{craig1957linear}, was carried over to conjunctions~$A \wedge B$
in \cite{DBLP:conf/cav/McMillan03}, and generalised to inductive
sequences of interpolants
in~\cite{Henzinger:2004,DBLP:conf/cav/McMillan06}. The concept of tree
interpolation, strictly generalising inductive sequences of
interpolants, is presented in the documentation of the interpolation
engine iZ3~\cite{tree-interpolants}; the computation of tree
interpolants by computing a sequence of binary interpolants is also
described in \cite{DBLP:conf/popl/HeizmannHP10}. In this paper, we
present a new form of interpolation, \emph{disjunctive interpolation},
which is strictly more general than sequences of interpolants and tree
interpolants.  Our implementation supports Presburger arithmetic,
including divisibility constraints \cite{iprincess2011}, which is
rarely supported by existing tools, yet helpful in practice
\cite{HojjatETAL12AcceleratingInterpolants}.

A further generalisation of inductive sequences of interpolants are
restricted DAG interpolants~\cite{DBLP:conf/sas/AlbarghouthiGC12},
which also include disjunctiveness in the sense that multiple paths
through a program can be handled simultaneously.  Disjunctive
interpolants are incomparable in power to restricted DAG interpolants,
since the former does not handle interpolation problems in the form of
DAGs, while the latter does not subsume tree interpolation.  A
combination of the two kinds of interpolants (``disjunctive DAG
interpolation'') is strictly more powerful (and harder) than
disjunctive interpolation, see Sect.~\ref{sec:complexity} for a
complexity-theoretic analysis. We discuss techniques and heuristics to
practically handle shared sub-trees in disjunctive interpolation,
extending the benefits of DAG interpolation to recursive programs.

Inter-procedural \textbf{software model checking} with
interpolants has been an active area of research.  
In the context of
predicate abstraction, it has been discussed how well-scoped
invariants can be inferred~\cite{Henzinger:2004} in the presence of
function calls. 
Based on the concept of Horn clauses, a predicate
abstraction-based algorithm for bottom-up construction of function
summaries was presented in \cite{DBLP:conf/pldi/GrebenshchikovLPR12}.
Verification of programs with procedures
is described in \cite{DBLP:conf/popl/HeizmannHP10} (using
nested word automata) as well as in \cite{DBLP:conf/vmcai/AlbarghouthiGC12}.

The use of \textbf{Horn clauses} as intermediate representation for
verification was proposed in \cite{DBLP:conf/popl/GuptaPR11}, with the
verification of concurrent programs as main application. The
underlying procedure for solving sets of recursion-free Horn clauses,
over the combined theory of linear \emph{rational} arithmetic and
uninterpreted functions, was presented in
\cite{DBLP:conf/aplas/GuptaPR11}. A range of further applications of
Horn clauses, including inter-procedural model checking, was given in
\cite{DBLP:conf/pldi/GrebenshchikovLPR12}. Horn clauses are also
proposed as intermediate/exchange format for verification problems in
\cite{verificationAsSMT}, and are natively supported by the SMT solver
Z3~\cite{Moura:2008}.  Our paper extends this work by giving
general results about solvability and
computational complexity, independent of any particular calculus.
Our experiments are with linear \emph{integer} arithmetic, arguably
a more faithful model of discrete computation than rationals 
\cite{HojjatETAL12AcceleratingInterpolants}.


\section{Example: Verification of Recursive Predicates}
\label{sec:example}

We start by showing how our approach can verify programs encoded as
Horn clauses, by means of predicate abstraction and a theorem prover
for Presburger arithmetic.  
Fig.~\ref{fig:merge} shows an example of a system of Horn clauses,
generated by a straightforward length abstraction of a merge
operation that accepts two sorted lists and produces a new
one by merging them. Addition of an element increases the
resulting length ($Z$) by one whereas the processing
continues with one of the argument lists shorter. After
invoking such an operation, we wish to check whether it is
possible for the resulting length $Z$ to be more than the
sum of the lengths of the argument lists $X + Y$. In general,
we encode error conditions as Horn clauses with $\mathit{false}$ in their head,
and refer to such clauses as error clauses, although such clauses
do not have a special semantic status in our system.
When invoked with these clauses as input, our verification
tool automatically identifies that the definition of \textsf{merge}
as the predicate $X + Y - Z \ge 0$ gives a solution to these
Horn clauses. In terms of safety (partial correctness), this means that
the error condition cannot be reached.
\begin{figure}[tb]
\begin{lstlisting}
(1)  merge(X,Y,Z) $\expl$ X = 0 $\pand$ Y >= 0 $\pand$ Z = Y
(2)  merge(X,Y,Z) $\expl$ Y = 0 $\pand$ X >= 0 $\pand$ Z = X
(3)  merge(X,Y,Z) $\expl$ Y1 = Y - 1 $\pand$ merge(X, Y1, Z1) $\pand$ Z = Z1 + 1
(4)  merge(X,Y,Z) $\expl$ X1 = X - 1 $\pand$ merge(X1, Y, Z1) $\pand$ Z = Z1 + 1
(5)  false  $\expl$ merge(X,Y,Z) $\pand$ Z > X + Y
\end{lstlisting}
\vspace*{-3mm}
\caption{Horn Clauses Abstracting the Merge of Two Sorted Lists and an Assertion on Resulting Length.
Variables are universally quantified in each clause.\label{fig:merge}}

\medskip
\begin{lstlisting}
(1)   merge(X,Y,Z) $\expl$ X = 0 $\pand$ Y >= 0 $\pand$ Z = Y
(3')  merge1(X,Y,Z) $\expl$ Y1 = Y - 1 $\pand$ merge(X, Y1, Z1) $\pand$ Z = Z1 + 1
(4')  merge1(X,Y,Z) $\expl$ X1 = X - 1 $\pand$ merge(X1, Y, Z1) $\pand$ Z = Z1 + 1
(5')  false  $\expl$ merge1(X,Y,Z) $\pand$ Z > X + Y
\end{lstlisting}
\vspace*{-3mm}
\caption{Extended recursion-free approximation of the Horn clauses in Fig.~\ref{fig:merge}.}
\label{fig:extendedApprox}
\end{figure}

Our approach uses counterexample-driven refinement to perform
verification. 
In this example, the abstraction of Horn clauses starts with
a trivial set of predicates, containing only the predicate $\mathit{false}$,
which is assumed to be a valid approximation until proven otherwise.
Upon examining a clause that has a concrete satisfiable formula on the right-hand
side (e.g. $X = 0 \pand Y >= 0 \pand Z = Y$), we rule out $\mathit{false}$ as the
approximation of \textsf{merge}. In the absence of other candidate predicates,
the approximation of \textsf{merge} becomes the conjunction of an empty set of
predicates, which is $\mathit{true}$.  Using this approximation the error clause
is no longer satisfied.
At this point the algorithm checks whether a true error is reached by
directly chaining the clauses involved in computing the approximation
of predicates. This amounts to checking whether the following recursion-free
subset of clauses has a solution:
\begin{lstlisting}
(1) merge(X,Y,Z) $\expl$ X = 0 $\pand$ Y >= 0 $\pand$ Z = Y
(5) false  $\expl$ merge(X,Y,Z) $\pand$ Z > X + Y
\end{lstlisting}
The solution to above problem is any formula $I(X,Y,Z)$ such that
\begin{lstlisting}
I(X,Y,Z) $\expl$ X = 0 $\pand$ Y >= 0 $\pand$ Z = Y
false  $\expl$ I(X,Y,Z) $\pand$ Z > X + Y
\end{lstlisting}
This is precisely an interpolant of $X = 0 \pand Y >= 0 \pand Z = Y$
and $Z > X + Y$. If our algorithm picks the interpolant $Z \le X + Y$,
the subsequent check shows it to be a solution and the program is successfully verified
(this is what happens in our current implementation).
In general, however, there is no guarantee about which of the interpolants will be picked,
so another valid solution is $P_1(X,Y,Z) \equiv Z = Y \land X \ge 0$. For illustration
purposes, suppose $P_1$ is the interpolant picked. The currently considered
possible contradiction for Horn clauses is thereby eliminated, and $P_1$ is added into
a list of abstraction predicates for the relation~\textsf{merge}.
Because the predicates approximating \textsf{merge} are now updated, we consider
the abstraction of the system in terms of these predicates. Because of the clause (2), however,
$P_1$ is not a conjunct in a valid approximation, which leads us to consider clauses (2) and (5)
and add, for example, $P_2(X,Y,Z) \equiv Z=X \land Y \ge 0$ as another predicate in
the approximation of \textsf{merge}. Note, however, that both $P_1$ and $P_2$
are ruled out as approximation of clause (3), so the following recursion-free
unfolding is not solved by the approximation so far:
\begin{lstlisting}
(1)   merge(X,Y,Z) $\expl$ X = 0 $\pand$ Y >= 0 $\pand$ Z = Y
(3')  merge1(X,Y,Z) $\expl$ Y1 = Y - 1 $\pand$ merge(X, Y1, Z1) $\pand$ Z = Z1 + 1
(5')  false  $\expl$ merge1(X,Y,Z) $\pand$ Z > X + Y
\end{lstlisting}
This particular problem could be reduced to solving an interpolation sequence, but it is
more natural to think of it simply as a solution for recursion-free Horn clauses. A solution
is an interpretation of the relations \textsf{merge} and \textsf{merge1} as ternary
relations on integers, such that the clauses are true. Note that this problem could also
be viewed as the computation of tree interpolants, which are also a special case of solving
recursion-free Horn clauses, as are DAG interpolants and a new notion of disjunctive
tree interpolants that we introduce. The general message, in line with 
\cite{DBLP:conf/popl/GuptaPR11,DBLP:conf/aplas/GuptaPR11,DBLP:conf/pldi/GrebenshchikovLPR12} 
is that
recursion-free clauses are a perfect fit for counterexample-driven verification: they allow
us to provide the theorem proving procedure with much more information that they can use to
refine abstractions. In fact, we could also provide further recursion-free approximations, such as
in Fig.~\ref{fig:extendedApprox}.
In the limit, the original set of clauses or its recursive unfoldings are its own approximations,
some of them exact, 
but the advantage of \emph{recursion-free} Horn clauses is that their solvability is decidable
under very general conditions. This provides us with a solid
theorem proving building block to construct robust and predictable
solvers for the undecidable recursive case. Our paper describes a new such building block: 
disjunctive interpolants, which correspond to a subclass of non-recursive Horn clauses.


\section{Formulae and Horn Clauses}

\paragraph{\ConstraintLangs.}
Throughout this paper, we assume that a first-order vocabulary of
\emph{interpreted symbols} has been fixed, consisting of a set~$\cal
F$ of fixed-arity function symbols, and a set~$\cal P$ of fixed-arity
predicate symbols. Interpretation of $\cal F$ and $\cal P$ is
determined by a class~${\cal S}$ of structures $(U, I)$
consisting of non-empty universe~$U$, and a mapping~$I$ that assigns
to each function in $\cal F$ a set-theoretic function over $U$, and to
each predicate in $\cal P$ a set-theoretic relation over $U$. As a
convention, we assume the presence of an equation symbol~``$=$'' in $\cal
P$, with the usual interpretation.  Given a countably infinite
set~$\cal X$ of variables, a \emph{\constraintLang} is a
set~\Con\ of first-order formulae over $\cal F, P, X$
For example, the language of quantifier-free Presburger arithmetic has
${\cal F} = \{+, -, 0, 1, 2, \ldots\}$ and ${\cal P} = \{=, \leq,
|\}$).

A \constraint\ is called \emph{satisfiable} if it holds for some
structure in $\cal S$ and some assignment of the variables~$\cal X$,
otherwise \emph{unsatisfiable}. We say that a set~$\Gamma \subseteq \Con$ of
\constraints\ \emph{entails} a \constraint~$\phi \in \Con$ if
every structure and variable assignment that satisfies all
\constraints\ in $\Gamma$ also satisfies $\phi$; this is denoted by
$\Gamma \models \phi$.

\fv{\phi} denotes the set of free variables in \constraint~$\phi$.
We write $\phi[x_1, \ldots, x_n]$ to state that a \constraint\
contains (only) the free variables $x_1, \ldots, x_n$, and $\phi[t_1,
\ldots, t_n]$ for the result of substituting the terms~$t_1, \ldots,
t_n$ for $x_1, \ldots, x_n$.  Given a \constraint~$\phi$ containing
the free variables~$x_1, \ldots, x_n$, we write~$\unicl (\phi)$ for
the \emph{universal closure} $\forall x_1, \ldots, x_n. \phi$.

\paragraph{Positions.}
We denote the set of \emph{positions} in a \constraint~$\phi$ by
\positions{\phi}. For instance, the \constraint~$a \wedge \neg a$ has
4~positions, corresponding to the sub-formulae~$a \wedge \neg a, \neg
a$, and the two occurrences of $a$. The sub-formula of a
formula~$\phi$ underneath a position~$p$ is denoted by
\under{\phi}{p}, and we write $\phi[p/\psi]$ for the result of
replacing the sub-formula~\under{\phi}{p} with $\psi$.  Further, we
write $p \leq q$ if position $p$ is above $q$ (that is, $q$ denotes a
position within the sub-formula~\under{\phi}{p}), and $p < q$ if $p$
is strictly above $q$.

\paragraph{Craig interpolation} is the main technique used to
construct and refine abstractions in software model checking.  A
binary interpolation problem is a conjunction~$A \wedge B$ of
constraints. A \emph{Craig interpolant} is a \constraint~$I$ such that
$A \models I$ and $B \models \neg I$, and such that $\fv{I} \subseteq
\fv{A} \cap \fv{B}$. The existence of an interpolant implies that $A
\wedge B$ is unsatisfiable. We say that a \constraintLang\ has the
\emph{interpolation property} if also the opposite holds: whenever $A
\wedge B$ is unsatisfiable, there is an interpolant~$I$.

\subsection{Horn Clauses}

To define the concept of Horn clauses, we fix a set~$\cal R$ of
uninterpreted fixed-arity \emph{\relsyms,} disjoint from $\cal P$ and $\cal F$. A \emph{Horn clause} is a
formula~$C \wedge B_1 \wedge \cdots \wedge B_n \to H$
where
\begin{itemize}
\item $C$ is a \constraint\ over $\cal F, P, X$;
\item each $B_i$ is an application~$p(t_1, \ldots, t_k)$ of a \relsym\
  $p \in \cal R$ to first-order terms over $\cal F, X$;
\item $H$ is similarly either an
  application~$p(t_1, \ldots, t_k)$ of $p \in \cal R$ to first-order terms,
  or is the \constraint\ $\mathit{false}$.
\end{itemize}
$H$ is called the \emph{head} of the clause, $C \wedge B_1 \wedge
\cdots \wedge B_n$ the \emph{body.} In case $C = \mathit{true}$, we
usually leave out $C$ and just write $B_1 \wedge \cdots \wedge B_n \to
H$. First-order variables (from $\cal X$) in a clause 
are considered implicitly universally quantified; \relsyms\ represent
set-theoretic relations over the universe~$U$ of a structure~$(U, I)
\in \cal S$. Notions like (un)satisfiability and entailment
generalise straightforwardly to formulae with \relsyms.

A \emph{\relsym\ assignment} is a mapping~$\mathit{sol} : {\cal R} \to
\Con$ that maps each $n$-ary \relsym~$p \in \cal R$ to a
\constraint~$\mathit{sol}(p) = C_p[x_1, \ldots, x_n]$ with $n$ free
variables. The \emph{instantiation}~$\mathit{sol}(h)$ of a Horn
clause~$h$ is defined by:
\begin{align*}
  \mathit{sol}\big(C \wedge p_1(\bar t_1) \wedge \cdots \wedge p_n(\bar t_n)
         \to p(\bar t)\big) &~=~
  C \wedge \mathit{sol}(p_1)[\bar t_1] \wedge \cdots \wedge
    \mathit{sol}(p_n)[\bar t_n] \to \mathit{sol}(p)[\bar t]
    \\
  \mathit{sol}\big(C \wedge p_1(\bar t_1) \wedge \cdots \wedge p_n(\bar t_n)
         \to \mathit{false}\big) &~=~
  C \wedge \mathit{sol}(p_1)[\bar t_1] \wedge \cdots \wedge
    \mathit{sol}(p_n)[\bar t_n] \to \mathit{false}
\end{align*}

\begin{definition}[Solvability]
  Let $\ClauseSet$ be a set of Horn clauses over \relsyms\ $\cal R$.
  \begin{enumerate}
  \item $\ClauseSet$ is called \emph{semantically solvable} if for
    every structure $(U, I) \in \cal S$ there is an interpretation of
    the \relsyms~$\cal R$ as set-theoretic relations over $U$ such
    the universally quantified closure~$\unicl(h)$ of every clause~$h \in \ClauseSet$ 
    holds in $(U,I)$.
  \item A $\ClauseSet$ is called \emph{syntactically solvable} if
    there is a \relsym\ assignment $\mathit{sol}$ such that for every
    structure $(U, I) \in \cal S$ and every clause~$h \in \ClauseSet$
    it is the case that $\unicl(\mathit{sol}(h))$ is satisfied.
  \end{enumerate}
\end{definition}

Note that, in the special case when $\cal S$ contains only
one structure, ${\cal S} = \{ (U,I) \}$, semantic
solvability reduces to the existence of relations
interpreting ${\cal R}$ that extend the structure $(U,I)$ in
such a way to make all clauses true. In
other words, Horn clauses are solvable in a structure if and
only if the extension of the theory of $(U,I)$ by relation symbols ${\cal R}$ 
in the vocabulary and by given Horn clauses as axioms is consistent.

Clearly, if a set of Horn clauses is syntactically solvable, then it
is also semantically solvable. The converse is not true in general,
because the solution need not be expressible in the constraint language
(see Appendix~\ref{example:nonlin} for an example).

A set~$\ClauseSet$ of Horn clauses induces a \emph{dependence
  relation}~$\to_{\ClauseSet}$ on $\cal R$, defining $p \to_{\ClauseSet} q$ if
there is a Horn clause in $\ClauseSet$ that contains $p$ in its head, and
$q$ in the body. The set~$\ClauseSet$ is called \emph{recursion-free} if
$\to_{\ClauseSet}$ is acyclic, and \emph{recursive} otherwise.  In the
next sections we study the solvability problem for recursion-free Horn
clauses. This case is relevant, since solvers for recursion-free Horn
clauses form a main component of many general Horn-clause-based
verification
systems~\cite{DBLP:conf/popl/GuptaPR11,DBLP:conf/pldi/GrebenshchikovLPR12}.


\section{Disjunctive Interpolants and Body-Disjoint Horn Clauses}
\label{sec:disjInts}

Having defined the classical notions of interpolation and Horn clauses,
we now present our notion of disjunctive interpolants, and the
corresponding class of Horn clauses.
Our inspiration are generalized forms of Craig interpolation, such as
inductive sequences of
interpolants~\cite{Henzinger:2004,DBLP:conf/cav/McMillan06} or tree
interpolants~\cite{tree-interpolants,DBLP:conf/popl/HeizmannHP10}.  We
introduce disjunctive interpolation as a new form of
interpolation that is tailored to the refinement of abstractions in
Horn clause verification, strictly generalising both inductive
sequences of interpolants and tree interpolation.  Disjunctive
interpolation problems can specify both conjunctive and disjunctive
relationships between interpolants, and are thus applicable for
simultaneous analysis of multiple paths in a program, but also
tailored to inter-procedural analysis or verification of concurrent
programs~\cite{DBLP:conf/pldi/GrebenshchikovLPR12}.

Disjunctive interpolation problems correspond to a specific fragment
of recursion-free Horn clauses, namely recursion-free body-disjoint
Horn clauses (see Sect.~\ref{sec:bodyDisjointClauses}).  The
definition of disjunctive interpolation is chosen deliberately to be
as general as possible, while still avoiding the high computational
complexity of solving general systems of recursion-free Horn
clauses. Computational complexity is discussed in
Sect.~\ref{sec:complexity}.

We introduce disjunctive interpolants as a means of \emph{sub-formula
  abstraction}. For example, given an unsatisfiable \constraint~$\phi[\alpha]$
containing $\alpha$ as a sub-formula in a positive position, the goal is to
find an abstraction~$\alpha'$ such that $\alpha \models \alpha'$ and
$\alpha[\alpha'] \models \mathit{false}$, and such that $\alpha'$ only
contains variables common to $\alpha$ and $\phi[\mathit{true}]$. Generalizing
this to any number of subformulas, we obtain the following.

\begin{definition}[Disjunctive interpolant]
  \label{def:disjInt}
  Let $\phi$ be a constraint, and $\mathit{pos} \subseteq
  \positions{\phi}$ a set of positions in $\phi$ that are only
  underneath the connectives~$\wedge$ and $\vee$. 
  A \emph{disjunctive interpolant} is a map~$I : \mathit{pos} \to \Con$ from
  positions to constraints such that:
  \begin{enumerate}
  \item For each position~$p \in \mathit{pos}$, with direct children\\
    $\{q_1, \ldots, q_n\}
      ~=~ \{q \in \mathit{pos} \mid p < q \text{~and~} \lnot\exists r
      \in \mathit{pos}.\; p < r < q\}$ \ we have
    \begin{equation*}
      \under{\big(\phi[q_1/I(q_1), \ldots, q_n/I(q_n)]\big)}{p} ~\models~ I(p)~,
    \end{equation*}
  \item For the topmost positions~$\{q_1, \ldots, q_n\} ~=~ \{q \in
    \mathit{pos} \mid \lnot\exists r \in \mathit{pos}.\; r < q\}$ \ we have
    \begin{equation*}
      \phi[q_1/I(q_1), \ldots, q_n/I(q_n)] ~\models~ \mathit{false}~,
    \end{equation*}
  \item For each position~$p \in \mathit{pos}$, \ we have \ \ $\fv{I(p)}
    \subseteq \fv{\under{\phi}{p}} \cap
    \fv{\phi[p/\mathit{true}]}$.
  \end{enumerate}
\end{definition}

\begin{example}
  Consider $A_p \wedge B$, with position~$p$ pointing to the
  sub-formula~$A$, and $\mathit{pos} = \{p\}$. The disjunctive
  interpolants for $A \wedge B$ and $\mathit{pos}$ coincide with the
  ordinary binary interpolants for $A \wedge B$.
\end{example}

\begin{example}
  Consider the formula $\phi =
  \big(\cdots\big(\big(\big(T_1\big)_{p_1} \wedge T_2\big)_{p_2}
  \wedge T_3\big)_{p_3} \wedge \cdots \big)_{p_{n-1}} \wedge T_n$ and
  positions~$\mathit{pos} = \{p_1, \ldots, p_{n-1}\}$. Disjunctive
  interpolants for $\phi$ and $\mathit{pos}$ correspond to inductive
  sequences of
  interpolants~\cite{Henzinger:2004,DBLP:conf/cav/McMillan06}. Note
  that we have the entailments\\
  $T_1 \models I(p_1)$,~
    $I(p_1) \wedge T_2 \models I(p_2)$,~ \ldots,~
    $I(p_{n-1}) \wedge T_n \models \mathit{false}$.
\end{example}

\begin{example}
  \label{ex:tree-interpolants}
  A tree interpolation
  problem~\cite{tree-interpolants,DBLP:conf/popl/HeizmannHP10} is
  given by a finite directed tree~$(V, E)$, writing $E(v, v')$ to
  express that the node~$v'$ is a direct child of $v$, together with a
  function $\phi : V \to \Con$ that labels each node~$v$ of the tree
  with a \constraint~$\phi(v)$. A \emph{tree interpolant} is a
  function~$I : V \to \Con$ such that
  \begin{inparaenum}
  \item $I(v_0) = \mathit{false}$ for the root node $v_0 \in V$,
  \item for any node~$v \in V$, the entailment $ \phi(v) \wedge
    \bigwedge_{(v, w) \in E} I(w) \models I(v), $ holds, and
  \item for any node~$v \in V$, every variable in $I(v)$ occurs both
    in some formula~$\phi(w)$ for $w$ such that $E^*(v, w)$, and in
    some formula~$\phi(w')$ for some $w'$ such that $\neg E^*(v,
    w')$. ($E^*$ is the reflexive transitive close of $E$).
  \end{inparaenum}

  It can be shown that a tree interpolant~$I$ exists if and only if
  $\bigwedge_{v \in V} \phi(v)$ is unsatisfiable.
  Tree interpolation
  problems~\cite{tree-interpolants,DBLP:conf/popl/HeizmannHP10}
  correspond to disjunctive interpolation with a set~$\mathit{pos}$ of
  positions that are only underneath~$\wedge$ (and never underneath
  $\vee$).
\end{example}

\begin{example}
  \label{ex:mergeExpansion}
  We consider the example given in Fig.~\ref{fig:extendedApprox},
  Sect.~\ref{sec:example}. To compute a solution for the Horn
  clauses, we first \emph{expand} the Horn clauses into a constraint,
  by means of exhaustive inlining (see Sect.~\ref{sec:rf-horn}), 
  obtaining a disjunctive interpolation problem:
  \begin{align*}
    \mathit{false} & ~~\leadsto~~
    \textsf{merge1}(X,Y,Z) \wedge Z > X + Y
    \\
    & ~~\leadsto~~
    \left(
      \begin{array}{c}
        Y1 = Y - 1 \wedge \textsf{merge}(X, Y1, Z1) \wedge Z = Z1 + 1\\ \vee\\
        X1 = X - 1 \wedge \textsf{merge}(X1, Y, Z1) \wedge Z = Z1 + 1
      \end{array}
    \right) \wedge Z > X + Y
    \\ &  ~~\leadsto~~
    \left(
      \begin{array}{c}
        Y1 = Y - 1 \wedge (X = 0 \wedge Y1 >= 0 \wedge Z1 = Y1)_{q}  \wedge Z = Z1 + 1\\ \vee\\
        X1 = X - 1 \wedge (X1 = 0 \wedge Y >= 0 \wedge Z1 = Y)_{r} \wedge Z = Z1 + 1
      \end{array}
    \right)_p \wedge Z > X + Y
  \end{align*}
  In the last formula, the positions~$p, q, r$ corresponding to the
  \relsym\ \textsf{merge1} and the two occurrences of \textsf{merge}
  are marked. It can be observed that the last formula is
  unsatisfiable, and that~$I = \{ p \mapsto X + Y \geq Z,\; q \mapsto
  X + Y1 \geq Z1,\; r \mapsto X1 + Y \geq Z1\}$ is a disjunctive
  interpolant. A solution for the Horn clauses can be derived from the
  interpolant by conjoining the constraints derived for the two
  occurrences of \textsf{merge}:
  \begin{equation*}
    \textsf{merge1}(X, Y, Z) = X + Y \geq Z,\quad
    \textsf{merge}(X, Y, Z) \begin{array}[t]{l}
                            = X + Y \geq Z \wedge X + Y \geq Z \\
                            = X + Y \geq Z
                            \end{array}
  \end{equation*}
\end{example}

\begin{theorem}
  \label{thm:disjInts}
  Suppose $\phi$ is a constraint, and suppose $\mathit{pos} \subseteq
  \positions{\phi}$ is a set of positions in $\phi$ that are only
  underneath the connectives~$\wedge$ and $\vee$. If \Con\ is a
  \constraintLang\ that has the interpolation property, then a
  disjunctive interpolant~$I$ exists for $\phi$ and $\mathit{pos}$ if
  and only if $\phi$ is unsatisfiable.
\end{theorem}

\begin{proof}
  ``$\Rightarrow$'' By means of simple induction, we can derive that
  $\under{\phi}{p} \models I(p)$ holds for every disjunctive
  interpolant~$I$ for $\phi$ and $\mathit{pos}$, and for every $p \in
  \mathit{pos}$. From Def.~\ref{def:disjInt}, it then follows that
  $\phi$ is unsatisfiable.

  ``$\Leftarrow$'' Suppose $\phi$ is unsatisfiable. We encode the
  disjunctive interpolation problem into a (conjunctive) tree
  interpolation
  problem~\cite{tree-interpolants,DBLP:conf/popl/HeizmannHP10} (also
  see Example~\ref{ex:tree-interpolants}) by adding auxiliary Boolean
  variables.  Wlog, we assume that $\mathit{pos}$ contains the root
  position~$\mathit{root}$ of $\phi$. The graph of the tree
  interpolation problem is $(\mathit{pos}, E)$, with the edge relation
  $E = \{(p, q) \mid p < q \text{~and~} \lnot\exists r. p < r < q
  \}$. For every $p \in \mathit{pos}$, let $a_p$ be a fresh Boolean
  variable. We label the nodes of the tree using the function~$\phi_L
  : \mathit{pos} \to \Con$. For each position~$p \in \mathit{pos}$,
  with direct children $\{q_1, \ldots, q_n\} ~=~ \{q \in \mathit{pos}
  \mid E(p, q)\}$ we define
  \begin{equation*}
    \phi_L(p) ~=~
    \begin{cases}
      \phi[q_1/a_{q_1}, \ldots, q_n/a_{q_n}] &
      \text{if~} p = \mathit{root}
      \\
      \neg a_p \vee \under{\big(\phi[q_1/a_{q_1}, \ldots, q_n/a_{q_n}]\big)}{p}
      & \text{otherwise}
    \end{cases}
  \end{equation*}
  Observe that $\bigwedge_{p \in \mathit{pos}} \phi_L(p)$ is
  unsatisfiable. As explained in Example~\ref{ex:tree-interpolants}, a
  tree interpolant~$I_T$ exists for this labelling function. By
  construction, for non-root positions~$p \in \mathit{pos} \setminus
  \{\mathit{root}\}$ the interpolant labelling is equivalent to
  $I_T(p) \equiv \neg a_p \vee I_p$, where $I_p$ does not contain any
  further auxiliary Boolean variables. We can then construct a
  disjunctive interpolant~$I$ for the original problem as
  \begin{equation*}
    I(p) ~=~
    \begin{cases}
      \mathit{false} & \text{if~} p = \mathit{root}
      \\
      I_p & \text{otherwise}
    \end{cases}
  \end{equation*}
  To see that $I$ is a disjunctive interpolant, observe that for each
  position~$p \in \mathit{pos}$ with direct children $\{q_1, \ldots,
  q_n\} ~=~ \{q \in \mathit{pos} \mid E(p, q)\}$ the following
  entailment holds (since $I_T$ is a tree interpolant):\quad
  $\phi_L(p) \wedge
  (\neg a_{q_1} \vee I_{q_1}) \wedge \cdots \wedge (\neg a_{q_n} \vee I_{q_n})
  ~\models~
  I_T(p)$
  \\
  Via Boolean reasoning this implies:\quad
  $\under{\big(\phi[q_1/I_{q_1}, \ldots, q_n/I_{q_n}]\big)}{p} ~\models~ I(p)$.
  \qed
\end{proof}

\subsection{Solvability of Body-Disjoint Horn Clauses}
\label{sec:bodyDisjointClauses}

The relationship between Craig interpolation and (syntactic) solutions
of Horn clauses has been observed in
\cite{DBLP:conf/aplas/GuptaPR11}. Disjunctive interpolation
corresponds to a specific class of recursion-free Horn clauses, namely
Horn clauses that are \emph{body disjoint:}
\begin{definition}
  A finite, recursion-free set~$\ClauseSet$ of Horn clauses is
  \emph{body disjoint} if for each \relsym~$p$ there is at most one
  clause containing $p$ in its body, and every clause
  contains $p$ at most once.
\end{definition}
An example for body-disjoint clauses is the subset~$\{(1), (2), (5)\}$
of clauses in Fig.~\ref{fig:merge}. Syntactic solutions of a
set~$\ClauseSet$ of body-disjoint Horn clauses can be computed by
solving a disjunctive interpolation problem; vice versa, every
disjunctive interpolation problem can be translated into an equivalent
set of body-disjoint clauses.

In order to extract an interpolation problem from $\ClauseSet$, we
first normalise the clauses: for every \relsym\ $p \in
\cal R$, we fix a unique vector of variables~$\bar x_p$, and rewrite
$\ClauseSet$ such that $p$ only occurs in the form~$p(\bar x_p)$. This
is possible due to the fact that $\ClauseSet$ is body disjoint.  The
translation from Horn clauses to a disjunctive interpolation problem
is done recursively, similar in spirit to inlining of function
invocations in a program; thanks to body-disjointness, the encoding is
polynomial.
\begin{align*}
  \mathit{enc}\big({\ClauseSet}\big) &~=~
  \bigvee_{( C \wedge B_1 \wedge \cdots \wedge B_n \to \mathit{false}) \,\in {\ClauseSet}}
  C \wedge \mathit{enc}'(B_1) \wedge \cdots \wedge \mathit{enc}'(B_n)
  \\
  \mathit{enc}'\big(p(\bar x_p)\big) &~=~
  \left(
  \bigvee_{( C \wedge B_1 \wedge \cdots \wedge B_n \to p(\bar x_p)) \,\in {\ClauseSet}}
  C \wedge \mathit{enc}'(B_1) \wedge \cdots \wedge \mathit{enc}'(B_n)
  \right)_{l_p}
\end{align*}
Note that the resulting formula~$\mathit{enc}(\ClauseSet)$ contains a
unique position~$l_p$ at which the definition of a \relsym~$p$ is
inlined; in the second equation, this position is marked with
$l_p$. Any disjunctive interpolant~$I$ for this set of positions
represents a syntactic solution of $\ClauseSet$, and vice versa.


\section{Solvability of Recursion-free Horn Clauses}
\label{sec:rf-horn}

The previous section discussed how the class of recursion-free
body-disjoint Horn clauses can be solved by reduction to disjunctive
interpolation. We next show that this construction can be generalised
to arbitrary systems of recursion-free Horn clauses. In absence of the
body-disjointness condition, however, the encoding of Horn clauses as
interpolation problems can incur a potentially exponential blowup. We
give a complexity-theoretic argument justifying that this blowup
cannot be avoided in general. This puts disjunctive interpolation (and,
equivalently, body-disjoint Horn clauses) at a sweet spot: preserving
the relatively low complexity of ordinary binary Craig
interpolation, while carrying much of the flexibility of the Horn
clause framework.

We first introduce the exhaustive
\emph{expansion}~$\mathit{exp}(\ClauseSet)$ of a set~$\ClauseSet$ of
Horn clauses, which generalises the Horn clause encoding from the
previous section. We write $C' \wedge B_1' \wedge \cdots \wedge B_n'
\to H'$ for a fresh variant of a Horn clause $C \wedge B_1 \wedge
\cdots \wedge B_n \to H$, i.e., the clause obtained by replacing all
free first-order variables with fresh variables. Expansion is then
defined by the following recursive functions:
\begin{align*}
  \mathit{exp}\big({\ClauseSet}\big) &~=~
  \bigvee_{( C \wedge B_1 \wedge \cdots \wedge B_n \to \mathit{false}) \,\in {\ClauseSet}}
  C' \wedge \mathit{exp}'(B_1') \wedge \cdots \wedge \mathit{exp}'(B_n')
  \\
  \mathit{exp}'\big(p(\bar t)\big) &~=~
  \bigvee_{( C \wedge B_1 \wedge \cdots \wedge B_n \to p(\bar s)) \,\in {\ClauseSet}}
  C' \wedge \mathit{exp}'(B_1') \wedge \cdots \wedge \mathit{exp}'(B_n')
  \wedge \bar t = \bar s'
\end{align*}
Note that $\mathit{exp}$ is only well-defined for finite and
recursion-free sets of Horn clauses, since the expansion might not
terminate otherwise.

\begin{theorem}[Solvability of recursion-free Horn clauses]
  \label{thm:solvability}
  Let $\ClauseSet$ be a finite, recursion-free set of Horn clauses. If the
  underlying \constraintLang\ has the interpolation property, then the
  following statements are equivalent:
  \begin{enumerate}
  \item $\ClauseSet$ is semantically solvable;
  \item $\ClauseSet$ is syntactically solvable;
  \item $\mathit{exp}({\ClauseSet})$ is unsatisfiable.
  \end{enumerate}
\end{theorem}

\begin{proof}
  $2 \Rightarrow 1$ holds because a syntactic solution gives rise to a
  semantic solution by interpreting the solution constraints.  $\neg 3
  \Rightarrow \neg 1$ holds because a model of
  $\mathit{exp}(\ClauseSet)$ witnesses domain elements that every
  semantic solution of $\ClauseSet$ has to contain, but which violate
  at least one clause of the form $C \wedge B_1 \wedge \cdots \wedge
  B_n \to \mathit{false}$, implying that no semantic solution can
  exist.  $3 \Rightarrow 2$ is shown by encoding $\ClauseSet$ into a
  disjunctive interpolation problem (Sect.~\ref{sec:disjInts}), which
  can solved with the help of Theorem~\ref{thm:disjInts}. To this end,
  clauses are first duplicated to obtain a problem that is
  body-disjoint, and subsequently normalised as described in
  Sect.~\ref{sec:bodyDisjointClauses}. More details are given in
  Appendix~\ref{app:rec-free-horns}.\qed
\end{proof}

\subsection{The Complexity of Recursion-free Horn Clauses}
\label{sec:complexity}

Theorem~\ref{thm:solvability} gives rise to a general algorithm for
(syntactically) solving recursion-free sets~$\ClauseSet$ of Horn clauses,
over \constraintLangs\ for which interpolation procedures are
available. The general algorithm requires, however, to generate and
solve the expansion~$\mathit{exp}(\ClauseSet)$ of the Horn clauses, which
can be exponentially bigger than $\ClauseSet$ (in case $\ClauseSet$ is not body disjoint), and might therefore require exponential
time. This leads to the question whether more efficient algorithms are
possible for solving Horn clauses.

We give a number of complexity results about (semantic) Horn clause
solvability; proofs of the results are given in
Appendix~\ref{app:PSPACE-hard}, \ref{app:s-exp}, and
\ref{app:NEXPTIME-hard}. Most importantly, we can observe that
solvability is PSPACE-hard, for every non-trivial constraint
language~\Con:\footnote{A similar observation was made in the
  introduction of \cite{DBLP:conf/cav/LalQL12}, for the case of
  programs with procedures.}
\begin{lemma}
  \label{lem:PSPACE-hard}
  Suppose a constraint language can distinguish at least two values,
  i.e., there are two ground terms~$t_0$ and $t_1$ such that $t_0
  \not= t_1$ is satisfiable. Then the semantic solvability problem for
  recursion-free Horn clauses is PSPACE-hard.
\end{lemma}

Looking for upper bounds, it is easy to see that solvability of Horn
clauses is in co-NEXPTIME for any \constraintLang\ with satisfiability
problem in NP (for instance, quantifier-free Presburger
arithmetic). This is because the size of the
expansion~$\mathit{exp}({\ClauseSet})$ is at most exponential in the
size of $\ClauseSet$. Individual \constraintLangs\ admit more
efficient solvability checks:
\begin{theorem}
  Semantic solvability of recursion-free Horn clauses over the
  \constraintLang\ of Booleans is PSPACE-complete.
\end{theorem}

\ConstraintLangs\ that are more expressive than the Booleans lead to a
significant increase in the complexity of solving Horn clauses. The 
lower bound in the following theorem can be shown by simulating time-bounded 
non-deterministic Turing machines.
\begin{theorem}
  \label{thm:co-NEXPTIME-complete}
  Semantic solvability of recursion-free Horn clauses over the
  \constraintLang\ of quantifier-free Presburger arithmetic is
  co-NEXPTIME-complete.
\end{theorem}

The lower bounds in Lemma~\ref{lem:PSPACE-hard} and
Theorem~\ref{thm:co-NEXPTIME-complete} hinge on the fact that sets of
Horn clauses can contain shared relation symbols in bodies. Neither
result holds if we restrict attention to body-disjoint Horn clauses,
which correspond to disjunctive interpolation as introduced in
Sect.~\ref{sec:disjInts}.  Since the
expansion~$\mathit{exp}(\ClauseSet)$ of body-disjoint Horn clauses is
linear in the size of the set of Horn clauses, also solvability can be
checked efficiently:
\begin{theorem}
  Semantic solvability of a set of body-disjoint Horn clauses, and
  equivalently the existence of a solution for a disjunctive
  interpolation problem, is in co-NP when working over the
  \constraintLangs\ of Booleans and quantifier-free Presburger
  arithmetic.
\end{theorem}
Body-disjoint Horn clauses are still expressive: they 
can directly encode acyclic control-flow graphs, 
as well as acyclic unfolding of many simple recursion patterns.

For proofs of all results of this section, please consult the Appendix.


\section{Model Checking with Recursive Horn Clauses}

Whereas \emph{recursion-free} Horn clauses generalise the concept
of Craig interpolation, solving \emph{recursive} Horn clauses corresponds
to the verification of general programs with loops, recursion, or
concurrency
features~\cite{DBLP:conf/pldi/GrebenshchikovLPR12}. Procedures to
solve recursion-free Horn clauses can serve as a building block
within model checking algorithms for recursive Horn
clauses~\cite{DBLP:conf/pldi/GrebenshchikovLPR12}, and are used to
construct or refine abstractions by analysing spurious
counterexamples. In particular, our disjunctive interpolation can be used for this
purpose, and offers a high degree of flexibility due to the
possibility to analyse counterexamples combining multiple execution
traces.  We illustrate the use of disjunctive interpolation within a
predicate abstraction-based algorithms for solving Horn clauses. Our
model checking algorithm is similar in spirit to the procedure in
\cite{DBLP:conf/pldi/GrebenshchikovLPR12}, and we explain it in
Sect.~\ref{sec:bottumUp}.

\paragraph{And/or trees of clauses.}

For sake of presentation, in our algorithm we represent
counterexamples (i.e., recursion-free sets of Horn clauses) in the
form of and/or trees labelled with clauses. Such trees are defined by
the following grammar:
\begin{align*}
  \mathit{AOTree} &~~::=~~ \mathit{And}(h, \mathit{AOTree},
  \ldots, \mathit{AOTree}) ~\mid~ \mathit{Or}(\mathit{AOTree},
  \ldots, \mathit{AOTree})
\end{align*}
where $h$ ranges over (possibly recursive) Horn clauses. We only
consider well-formed trees, in which the children of every
$\mathit{And}$-node have head symbols that are consistent with the
body literals of the clause stored in the node, and the sub-trees of
an $\mathit{Or}$-node all have the same head symbol.  And/or trees are
turned into body-disjoint recursion-free sets of clauses by renaming
\relsyms\ appropriately.
\begin{example}
  Referring to the clauses in Fig.~\ref{fig:merge}, a possible and/or
  tree is
  \begin{equation*}
    \mathit{And}\big((5),\; \mathit{And}\big((3),\;
           \mathit{Or}(\mathit{And}((1)),\; \mathit{And}((2)))\big)\big)
  \end{equation*}
  A corresponding set of body-disjoint recursion-free clauses is:
\begin{lstlisting}
(1')  merge2(X,Y,Z) $\expl$ X = 0 $\pand$ Y >= 0 $\pand$ Z = Y
(2')  merge2(X,Y,Z) $\expl$ Y = 0 $\pand$ X >= 0 $\pand$ Z = X
(3')  merge1(X,Y,Z) $\expl$ Y1 = Y - 1 $\pand$ merge2(X, Y1, Z1) $\pand$ Z = Z1 + 1
(5')  false  $\expl$ merge1(X,Y,Z) $\pand$ Z > X + Y
\end{lstlisting}
\end{example}

\paragraph{Solving and/or dags.}
Counterexamples extracted from model checking problems often assume
the form of and/or \emph{dags}, rather than and/or \emph{trees}. Since
and/or-dags correspond to Horn clauses that are not body-disjoint, the
complexity-theoretic results of the last section imply that it is in
general impossible to avoid the expansion of and/or-dags to
and/or-trees; there are, however, various effective techniques to
speed-up handling of and/or-dags (somewhat related to the techniques
in \cite{DBLP:conf/cav/LalQL12}). 
We highlight two of the techniques we use in our interpolation
engine Princess~\cite{iprincess2011}, which we used in our experimental
evaluation of the next section:

\emph{1) counterexample-guided expansion} expands
and/or-dags lazily, until an unsatisfiable fragment of the fully
expanded tree has been found; such a fragment is sufficient to compute
a solution. Counterexamples can determine which or-branch of an
and/or-dag is still satisfiable and has to be expanded
further.

\emph{2) and/or dag restructuring} factors out common
sub-dags underneath an $\mathit{Or}$-node, making the
and/or-dag more tree-like.

\subsection{A Predicate Abstraction-based Model Checking Algorithm}
\label{sec:bottumUp}

Our model checking algorithm is in Fig.~\ref{fig:BottomUpARG},
and similar in spirit as the procedure in
\cite{DBLP:conf/pldi/GrebenshchikovLPR12}; it has been implemented in
the model checker
Eldarica.\footnote{\url{http://lara.epfl.ch/w/eldarica}} Solutions for
Horn clauses are constructed in disjunctive normal form by building an
abstract reachability graph over a set of given predicates. When a
counterexample is detected (a clause with consistent body literals and
head~$\mathit{false}$), a theorem prover is used to verify that the
counterexample is genuine; spurious counterexamples are eliminated by
generating additional predicates by means of disjunctive interpolation.

In Fig.~\ref{fig:BottomUpARG}, $\Pi: {\cal R} \rightarrow 2^P$ denotes
a mapping from \relsyms\ to the current set of predicates used to
approximate the \relsym.
Given a (possibly recursive) set~$\ClauseSet$ of Horn clauses, we define
an \emph{abstract reachability graph} (ARG) as a hyper-graph~$(S, E)$,
where
\begin{itemize}
\item $S \subseteq \{ (p, Q) \mid p \in {\cal R}, Q \subseteq \Pi(p)
  \}$ is the set of nodes, each of which is a pair consisting of a
  \relsym\ and a set of predicates.
\item $E \subseteq S^* \times \ClauseSet \times S$ is the
  hyper-edge relation, with each edge being labelled with a
  clause. An edge $E(\langle s_1, \ldots, s_n\rangle, h, s)$, with
  $h = (C \wedge B_1 \wedge \cdots \wedge
  B_n \to H) \in \ClauseSet$, implies that
  \begin{itemize}
  \item $s_i = (p_i, Q_i)$ and $B_i = p_i(\bar t_i)$ for all $i = 1,
    \ldots, n$, and
  \item $s = (p, Q)$, $H = p(\bar t)$, and $Q = \{ \phi \in \Pi(p)
    \mid C \wedge Q_1[\bar t_1] \wedge \cdots \wedge Q_n[\bar t_n]
    \models \phi[\bar t] \}$, where we write $Q_i[\bar t_i]$ for the
    conjunction of the predicates~$Q_i$ instantiated for the argument
    terms~$t_i$.
  \end{itemize}
\end{itemize}

An ARG~$(S, E)$ is called \emph{closed} if the edge relation
represents all Horn clauses in $\ClauseSet$. This means, for every
clause~$h = (C \wedge p_1(\bar t_1) \wedge \cdots \wedge p_n(\bar t_n) \to
H) \in \ClauseSet$ and every sequence~$(p_1, Q_1), \ldots, (p_n, Q_n) \in
S$ of nodes one of the following properties holds:
\begin{itemize}
\item $C \wedge Q_1[\bar t_1] \wedge \cdots \wedge Q_n[\bar t_n] \models
  \mathit{false}$, or
\item there is an edge~$E(\langle (p_1, Q_1), \ldots, (p_n,
  Q_n)\rangle, C, s)$ such that $s = (p, Q)$, $H = p(\bar t)$, and $Q
  = \{ \phi \in \Pi(p) \mid C \wedge Q_1[\bar t_1] \wedge \cdots \wedge
  Q_n[\bar t_n] \models \phi[\bar t] \}$.
\end{itemize}

\begin{lemma}
  \label{lem:CEGAR-completeness}
  A set~$\ClauseSet$ of Horn clauses has a closed ARG $(S, E)$ if and only
  if $\ClauseSet$ is syntactically solvable.
\end{lemma}

A proof is given in Appendix~\ref{app:CEGAR-completeness}.  The
function~\textsc{ExtractCEX} (non-deterministically) extracts an
and/or-tree representing a set of counterexamples, which can be turned
into a recursion-free body-disjoint set of Horn clauses, and solved as
described in Sect.~\ref{sec:bodyDisjointClauses}. In general, the tree
contains both conjunctions (from clauses with multiple body literals)
and disjunctions, generated when following multiple hyper-edges (the
case $|T| > 1$). Disjunctions make it possible to eliminate multiple
counterexamples simultaneously. The algorithm is parametric in the
precise strategy used to compute counterexamples, choices evaluated in
the experiments are
\begin{description}
\item[TI] extraction of a single counterexamples with minimal depth\\
  (which means that disjunctive interpolation reduces to \textbf{T}ree
  \textbf{I}nterpolation), and
\item[DI] simultaneous extraction of all counterexamples with minimal
  depth\\ (so that genuine \textbf{D}isjunctive \textbf{I}nterpolation is used).
\end{description}

\begin{figure}[tb]
\begin{algorithmic}
\State $S:=\emptyset, ~ E:=\emptyset$ \Comment{Empty graph}
\Function{ConstructARG}{}
\While{$\mathit{true}$}
\State pick clause
   $h = (C \wedge p_1(\bar t_1) \wedge \cdots \wedge p_n(\bar t_n) \to
   H) \in \ClauseSet$
   \State\qquad and nodes $(p_1, Q_1), \ldots, (p_n, Q_n) \in S$
  \State\qquad such that
     $\neg\exists s.\; (\langle(p_1, Q_1), \ldots, (p_n, Q_n)\rangle, h, s) \in E$
  \State\qquad and $C \wedge Q_1[\bar t_1] \wedge \cdots \wedge Q_n[\bar t_n] \not\models \mathit{false}$

\smallskip
\If{no such clauses and nodes exist}\ \Return{$\ClauseSet$ is solvable}

\EndIf
\smallskip
\If{$H = \mathit{false}$}
  \Comment{Refinement needed}
  \State
    $\mathit{tree} := \mathit{And}(h, \textsc{ExtractCEX}(p_1, Q_1),
   \ldots, \textsc{ExtractCEX}(p_n, Q_n)$
   \If{$\mathit{tree}$ is unsatisfiable}
     \State extract disjunctive interpolant from $\mathit{tree}$, add predicates to $\Pi$
     \State delete part of $(S, E)$ used to construct $\mathit{tree}$
   \Else\ \Return{$\ClauseSet$ is unsolvable, with counterexample trace $\mathit{tree}$}
   \EndIf
\Else
\Comment{Add edge to ARG}
\State then $H = p(\bar t)$
\State $Q := \{ \phi \in \Pi(p) \mid \{C\} \cup Q_1 \cup \ldots \cup Q_n
  \models \phi \}$
\State $e := (\langle(p_1, Q_1), \ldots, (p_n, Q_n)\rangle, h, (p, Q))$
  \State $S := S \cup \{ (p, Q) \}, ~ E := E \cup \{ e \}$
\EndIf
\EndWhile
\EndFunction

\medskip
\Function{ExtractCEX}{$\mathit{root} : S$}
\Comment{Extract disjunctive interpolation problem}
\State pick $\emptyset \not= T \subseteq E$
 with $\forall e \in T.~ e = (\_, \_, \mathit{root})$
\State
 \Return{$\mathit{Or} \big\{\, \mathit{And}(h, \textsc{ExtractCEX}(s_1),
   \ldots, \textsc{ExtractCEX}(s_n)) \mid
     (\langle s_1, \ldots, s_n\rangle, h, \mathit{root}) \in T \,\big\}$}
\EndFunction
\end{algorithmic}

\vspace*{-2ex}
\caption{Algorithm for construction of abstract reachability graphs.
  \label{fig::bottomUpARG}}
\label{fig:BottomUpARG}
\end{figure}


We remark that we have also implemented a simpler ``global'' algorithm
(see Sect.~\ref{sec:example}), which
approximates each relation symbol globally with a single conjunction
of inferred predicates. In contrast, the above algorithm allows
multiple nodes, each of which contains a different conjunction, thus
allowing a disjunction of conjunctions of predicates. Both algorithms
behave similarly in our experience, with the global one occasionally slower,
but conceptually simpler.  Note that, what allowed us to use a simpler
algorithm at all is the fact that the interpolation problem considered
is more general. Thus, another advantage of more expressive
forms of interpolation is the simplicity of the resulting verification
algorithms built on top of them.


\subsection{Experimental Evaluation}\label{sec:experiments}

\begin{figure}[p]
\footnotesize
~~
\begin{tabular}[t]{lrrr}
\hline
\multirow{2}{*}{Model} & \multicolumn{3}{c}{ Time [s]} \\
& \textbf{TI} & ~~~~\textbf{DI} & ~~~~\textbf{HSF} \\
\hline
\multicolumn{4}{l}{\textbf{(a) Recursive Models }} \\
\hline
addition (C)      & 0.46	&	0.47& 0.19 	\\
binarysearch (C)  & 0.52	&	0.53& 0.16 	\\
mccarthy-91 (C)   & 1.19	&	1.01& 0.17 	\\
mccarthy-92 (C)~~~~~   & 10.43	&	5.90& ERR  	\\
palindrome  (C)   & 0.92	&	1.66& 0.16 	\\
remainder  (C)    & 0.77	&	0.76& ERR  	\\
identity  (C)     & 0.73	&	0.93& 0.16 	\\
merge (C)         & 0.76	&	1.50& 0.16 	\\
parity (C)        & 0.80	&	0.80& ERR  	\\
running  (C)      & 0.60	&	0.60& 0.15 	\\
triple (C)        & 1.92	&	1.96& ERR  	\\
\hline
\multicolumn{4}{l}{\textbf{(b) Examples from L2CA \cite{cav06}}} \\
\hline
bubblesort (E)    & 2.32	&	2.52& 2.51	\\
insdel (E)        & 0.38	&	0.35& 0.18	\\
insertsort (E)    & 1.31	&	1.47& 0.46	\\
listcounter (E)   & 0.44	&	0.48& 0.19	\\
listreversal (C)  & 1.28	&	1.34& 0.32	\\
listreversal (E)  & 1.55	&	1.54& 1.54	\\
mergesort (E)     & 0.64	&	0.78& 0.29	\\
selectionsort (E) & 1.72	&	1.46& 1.24	\\
\hline
\multicolumn{4}{l}{\textbf{(c) VHDL models from \cite{smrcka-vojnar08}}} \\
\hline
counter (C)       & 0.87	&	0.94&0.16	\\
register (C)      & 0.80	&	0.86&0.17	\\
\hline
\multicolumn{4}{l}{\textbf{(d) Verification conditions for}} \\
\multicolumn{4}{l}{\textbf{array programs \cite{cav09}}} \\
\hline
rotation\_vc.1 (C) & 2.13	&	2.58& 0.32	\\
rotation\_vc.2 (C) & 3.21	&	3.74& 0.34	\\
rotation\_vc.3 (C) & 3.36	&	3.29& 0.31	\\
rotation\_vc.1 (E) & 1.94	&	1.97& 0.30	\\
split\_vc.1 (C)    & 5.23	&	5.83& 0.97	\\
split\_vc.2 (C)    & -   	&	4.45& 0.69	\\
split\_vc.3 (C)    & -  	&	4.93& 0.65	\\
split\_vc.1 (E)    & 4.46	&	3.42& 1.11	\\
\hline
\multicolumn{4}{l}{\textbf{(e) NECLA benchmarks}} \\
\hline
inf1 (E) & 0.89	&	0.87& 0.24	\\
inf4 (E) & 1.07	&	1.53& 0.38	\\
inf6 (C) & 1.30	&	1.34& 0.19	\\
inf8 (C) & 1.73	&	1.89& 0.22	\\
\hline
\end{tabular}
\hspace*{\fill}%
\begin{tabular}[t]{lrrr}
\hline
\multirow{2}{*}{Model} & \multicolumn{3}{c}{ Time [s]} \\
& \textbf{TI} & ~~~~\textbf{DI} & ~~~~~~\textbf{HSF} \\
\hline
\multicolumn{4}{l}{\textbf{(f) Examples from \cite{ganty-rupak}}} \\
\hline
h1 (E)            & 8.32	&15.90  	  & 0.65	\\
h1.opt (E)        & 1.19	&1.00   	  & 0.26   \\
h1h2 (E)          & 16.49	&30.98  	  & 1.10	\\
h1h2.opt (E)      & 3.63	&2.00   	  & 0.33   \\
simple (E)        & 10.94	&16.80  	  & 1.07	\\
simple.opt (E)    & 1.23	&1.02   	  & 0.25   \\
test0 (C)         & 21.78	&110.11 	  & 1.57   \\
test0.opt (C)     & 2.98	&3.00   	  & 0.29   \\
test0 (E)          & 9.35	&35.42  	  & 0.62   \\
test0.opt (E)      & 1.14	&0.99   	  & 0.25   \\
test1.opt (C)      & 4.99	&10.84  		& 0.66   \\
test1 (E)          & 117.41	&364.48 	   & 102.88 \\
test1.opt (E)      & 4.33	&4.89   	     & 0.54   \\
test2\_1 (E)       & 55.67	&145.07 		& 189.31 \\
test2\_1.opt (E)   & 3.36	&3.40   		& 0.41   \\
test2\_2 (E)       & 145.79	&127.21 		& 132.55 \\
test2\_2.opt (E)   & 4.54	&4.52   		& 0.36   \\
test2.opt (C)      & 46.41	&163.41 		& 2.65   \\
wrpc.manual (C)~~~    & 0.55	&0.68   		& 0.17   \\
wrpc (E)           & 21.00	&26.73  		& 2.99   \\
wrpc.opt (E)       & 2.40	&2.35   		& 0.52   \\
\hline
\end{tabular}
~~

\hspace*{\fill}%
\raisebox{0ex}[0ex][0ex]{%
  \includegraphics[trim=7 2 5 -5,width=0.54\linewidth]{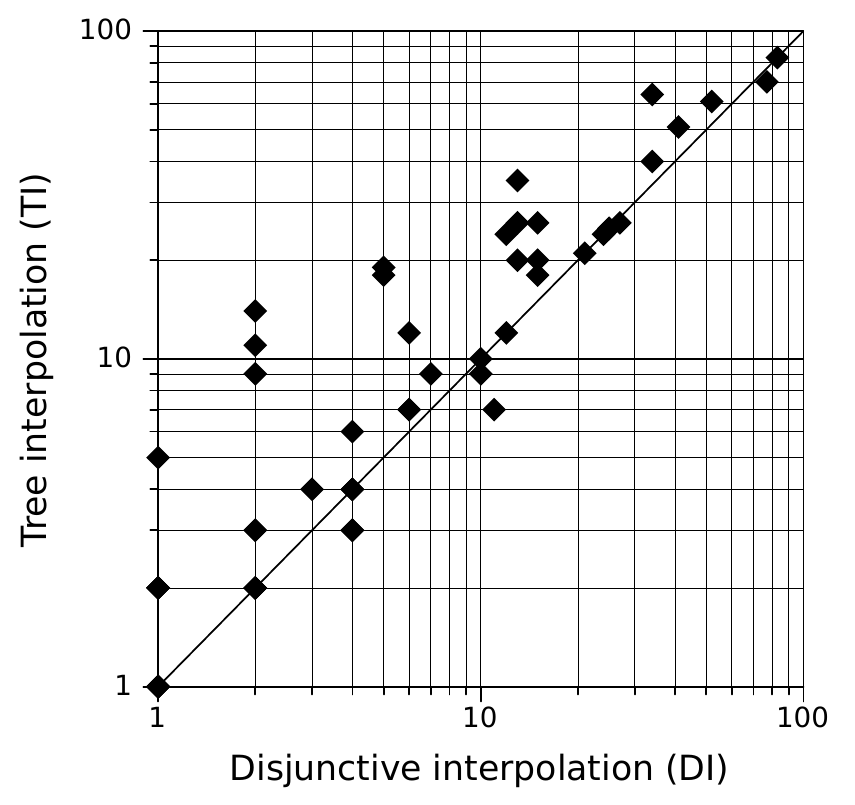}}%

\caption{Benchmarks for model checking Horn clauses. The letter
  after the model name distinguishes \textbf{C}orrect from models with
  a reachable \textbf{E}rror state. ``-'' indicates timeout.  The
  scatter plot illustrates the required number of refinement steps, for the
  case of single counterexamples (\textbf{TI}) and
  simultaneous extraction of all minimal-depth counterexamples
  (\textbf{DI}). All experiments were done on an
  Intel Core~i5 2-core machine with 3.2GHz and 8Gb, with a timeout of 500s.}
  \label{fig:data}
\end{figure}

We have evaluated our algorithm on a set of
benchmarks\footnote{\url{https://svn.sosy-lab.org/software/sv-benchmarks/trunk/clauses/LIA/}}
in integer linear arithmetic from the NTS library~\cite{fmtool2012}.
The (a) benchmarks are recursive algorithms, (b) benchmarks are
extracted from programs with singly-linked lists, (c) benchmarks are
models extracted from VHDL models of circuits, (d) benchmarks are
verification conditions for programs with arrays, (e) benchmarks are
in the NECLA static analysis suite, (f) C programs with asynchronous
procedure calls translated into NTS using the approach of
\cite{ganty-rupak} (the examples with extension .opt are obtained
via an optimised translation method [Pierre Ganty, personal
communication]. The results are given in Fig.~\ref{fig:data}.

The experiments show comparable verification times and
performance for the \textbf{T}ree \textbf{I}nterpolation and
\textbf{D}isjunctive \textbf{I}nterpolation runs. Studying the results
more closely, we observed that \textbf{DI} consistently led to a
smaller number of abstraction refinement steps (the scatter plot in
Fig.~\ref{fig:data}); this indicates that \textbf{DI} is indeed able
to eliminate multiple counterexamples simultaneously, and to rapidly
generate predicates that are useful for abstraction.  The experiments
also showed that there is a trade-off between the time spent
generating predicates, and the quality of the predicates. In
\textbf{TI}, on average $31\%$ of the verification is used for
predicate generation (interpolation), while with \textbf{DI} $42\%$ is
used; in some of the benchmarks in (f), this led to the phenomenon
that \textbf{DI} was slower than \textbf{TI}, despite fewer refinement
steps. We expect this will change as we make further improvements to
our prototypical implementation of disjunctive interpolation.

We compared our results to the performance of 
HSF,\footnote{\url{http://www7.in.tum.de/tools/hsf/}} a sophisticated
state-of-the-art verification engine for problems expressed as 
Horn clauses. We observe
similar performance on many benchmarks, with HSF notably faster on many
(f) benchmarks but the difference less pronounced for large benchmarks. 
We were
unable to process with HSF the benchmarks in (a) containing modular arithmetic; we marked those with ERR.


\subsection*{Conclusions}

We have introduced disjunctive interpolation as a new form of Craig
interpolation tailored to model checkers based on the paradigm of Horn
clauses. Disjunctive interpolation can be identified as solving
body-disjoint systems of recursion-free Horn clauses, and subsumes a
number of previous forms of interpolation, including tree
interpolation and inductive sequences of interpolants. We believe that
the flexibility of disjunctive interpolation is highly beneficial for
building interpolation-based model checkers. In particular, when
implementing more intelligent techniques (than used in our
experiments) to select sets of counterexamples handed over to
interpolation, significant speed-ups can be expected. We plan to
explore this direction in future work, together with 
improvements in the implementation of disjunctive interpolation
itself.


\clearpage

\appendix

\section{Solving Recursion-free Horn Clauses: Proof of
  Theorem~\ref{thm:solvability}}
\label{app:rec-free-horns}

We outline a proof for Theorem~\ref{thm:solvability}, direction $3
\Rightarrow 2$. Suppose the expansion $\mathit{exp}({\ClauseSet})$ of a
set~$\ClauseSet$ of recursion-free Horn clauses is unsatisfiable. As
before, we compute a solution of the Horn clauses separately for every
connected component of the $\to_{\ClauseSet}$-graph. Wlog we can therefore
assume that the $\to_{\ClauseSet}$-graph is connected.

\paragraph{Elimination of duplicated relation symbols.}

Furthermore, we can assume that every relation variable occurs at most
once in the body of a clause. Otherwise we duplicate the relation
variable (and all clauses defining it), and solve the resulting
simpler system. E.g., if we have clauses

\begin{equation*}
  p(x, y) \wedge p(y, z) \to r(x, z),\quad
  q(x, y) \to p(x, y),\quad
  x \geq 0 \to q(x, x)
\end{equation*}
we first expand the system to
\begin{multline*}
  p_1(x, y) \wedge p_2(y, z) \to r(x, z),\\
  q_1(x, y) \to p_1(x, y),\quad
  x \geq 0 \to q_1(x, x),\quad
  q_2(x, y) \to p_2(x, y),\quad
  x \geq 0 \to q_2(x, x)
\end{multline*}
and solve the expanded system. Afterwards we construct a solution of the
original system as
\begin{equation*}
  C_p[x,y] ~=~
  C_{p_1}[x,y] \wedge C_{p_2}[x,y],\quad
  C_q[x,y] ~=~
  C_{q_1}[x,y] \wedge C_{q_2}[x,y]
\end{equation*}
This is possible because the space of (syntactic) solutions of a Horn
clause is closed under conjunction.

\paragraph{Renaming of first-order variables and normalisation.}

We normalise the resulting clauses like in
Sect.~\ref{sec:bodyDisjointClauses}: for every \relsym\ $p$, we fix a unique
vector of variables~$\bar x_p$, and rewrite $\ClauseSet$ such that $p$
only occurs in the form~$p(\bar x_p)$; by renaming variables, we then
ensure that every variable~$x$ that is not argument of a \relsym\
occurs in at most one clause.

\paragraph{Encoding into a disjunctive interpolation problem.}

The translation from Horn clauses to a disjunctive interpolation
problem is done by adapting the expansion function~$\mathit{exp}$ from
Sect.~\ref{sec:rf-horn}:
\begin{align*}
  \mathit{enc}\big({\ClauseSet}\big) &~=~
  \bigvee_{( C \wedge B_1 \wedge \cdots \wedge B_n \to \mathit{false}) \,\in {\ClauseSet}}
  C' \wedge \mathit{enc}(B_1) \wedge \cdots \wedge \mathit{enc}(B_n)
  \\
  \mathit{enc}\big(p(\bar x_p)\big) &~=~
  \left(
  \bigvee_{( C \wedge B_1 \wedge \cdots \wedge B_n \to p(\bar x_p)) \,\in {\ClauseSet}}
  C' \wedge \mathit{enc}(B_1) \wedge \cdots \wedge \mathit{enc}(B_n)
  \right)_{l_p}
\end{align*}
Note that the resulting formula~$\mathit{enc}(\ClauseSet)$ contains a
unique position~$l_p$ at which the definition of a \relsym~$p$ is
inlined; in the second equation, this position is marked with
$l_p$. We then derive a disjunctive interpolant~$I$ for this set of
positions in $\mathit{enc}(\ClauseSet)$. A syntactic solution of $\ClauseSet$
is then given by the definition~$\forall \bar x_p. \big( p(\bar x_p)
\leftrightarrow I(l_p)\big)$, for all \relsyms~$p$.


\section{Solvability of Recursion-free Horn Clauses is PSPACE-hard:
  Proof of Lemma~\ref{lem:PSPACE-hard}}
\label{app:PSPACE-hard}

We reduce the unsatisfiability problem of quantified boolean formulae
(known to be PSPACE-hard) to solvability of recursion-free Horn
clauses.  Assume an arbitrary QBF of the shape~$\phi = Q_1 x_1. Q_2
x_2. ... Q_n x_n. F$, where $Q_i \in {\exists, \forall}$ are
quantifiers, $x_i$ are all variables occurring in the formula, and $F$
is a quantifier-free Boolean formula in CNF.

We translate $\phi$ into a recursion-free set of Horn clauses:
\begin{itemize}
\item a literal $x_i$ of a clause $C_j$ in $F$ becomes a Horn clause\\
  $x_i=t_1 \to C_{i,j}(x_1, x_2, \ldots, x_{i-1}, t_1, x_{i+1}, \ldots,
  x_n)$
\item a literal $\neg x_i$ of a clause $C_j$ in $F$ becomes a Horn
  clause\\
  $x_i=t_0 \to C_{i,j}(x_1, x_2, \ldots, x_{i-1}, t_0, x_{i+1},
  \ldots, x_n)$
\item a clause $C_j$ in $F$ becomes a set of Horn clauses\\
  $C_{1,j}(x_1, \ldots) \to C_j(x_1, \ldots), \quad
  C_{2,j}(x_1, \ldots) \to C_j(x_1, \ldots), \quad \ldots$
\item the body $F$ becomes the Horn clause\\
  $C_1(x_1, \ldots) \wedge
  C_2(x_1, \ldots) \wedge \cdots \to F_n(x_1, \ldots)$
\item a quantifier $Q_i = \exists$ is translated as the two clauses\\
  $F_{i+1}(x_1, \ldots, x_{i-1}, 0) \to F_i(x_1, \ldots, x_{i-1}),\quad
  F_{i+1}(x_1, \ldots, x_{i-1}, 1) \to F_i(x_1, \ldots, x_{i-1})$
\item a quantifier $Q_i = \forall$ is translated as the clause\\
  $F_{i+1}(x_1, \ldots, x_{i-1}, 0) \wedge F_{i+1}(x_1, \ldots,
  x_{i-1}, 1) \to F_i(x_1, \ldots, x_{i-1})$
\item finally, we add the clause $F_1() \wedge t_0 \not= t_1 \to
  \mathit{false}$.
\end{itemize}
It is now easy to see that the expansion~$\mathit{exp}(\ClauseSet)$ of the
Horn clauses coincides with the result of expanding all quantifiers in
$\phi$.
By Theorem~\ref{thm:solvability}, unsatisfiability of the expansion is
equivalent to solvability of the set of Horn clauses.

\section{Succinct Expansion of Recursion-free Horn Clauses}
\label{app:s-exp}

The following lemma implies that solvability of recursion-free Horn
clauses over the theory of Booleans is PSPACE-complete:
\begin{lemma}[Succinct expansion]
  Let $\ClauseSet$ be a finite, recursion-free set of Horn clauses. If the
  underlying \constraintLang\ provides quantifiers, in (deterministic)
  linear time a formula~$\mathit{sexp}({\ClauseSet})$ can be extracted
  that is equivalent to $\mathit{exp}({\ClauseSet})$. The number of
  quantifier alternations in $\mathit{sexp}({\ClauseSet})$ is at most two
  times the number of \relsyms\ in $\ClauseSet$.
\end{lemma}

\begin{proof}
  We assume that the Horn clauses are connected, i.e., the $\to_{\ClauseSet}$-graph 
  consists of a single connected component. Further, we
  assume that the first-order variables in any two clauses in $\ClauseSet$
  are disjoint. The encoding of Horn clause as a QBF formula is then
  defined by the following algorithm in pseudo-code. The algorithm
  maintains a list $\mathit{quantifiers}$ of quantifiers that have to
  be added in front of the formula.

\medskip
  \begin{algorithmic}
    \State $\mathit{quantifiers} \leftarrow \epsilon,\; \mathit{checksRequired}
    \leftarrow \emptyset$

\medskip
    \Function{Encode}{$\ClauseSet$}
    \State Order clauses~$\ClauseSet$ in topological order, starting
    from clauses with head~$\mathit{false}$
    \State $\mathit{matrix} \leftarrow
        \textsc{EncodeBodies}(\{C \wedge p_1(\bar t_1) \wedge \cdots \wedge p_n(\bar t_n) \to \mathit{false} \in \ClauseSet\}, \epsilon)$
    \State $\mathit{remaining} \leftarrow \{
                C \wedge p_1(\bar t_1) \wedge \cdots \wedge p_n(\bar t_n) \to p(\bar t) \in \ClauseSet\}$

    \While{$\mathit{remaining} \not= \emptyset$}
    \State Pick first clause
    $C \wedge p_1(\bar t_1) \wedge \cdots \wedge p_n(\bar t_n) \to p(\bar t) \in \ClauseSet$ in topological order
    \State $\mathit{nextClauses} \leftarrow \{c \in \ClauseSet \mid \text{head symbol
         of $c$ is $p$}\}$
    \State $\mathit{remaining} \leftarrow \mathit{remaining} \setminus \mathit{nextClauses}$
    \For{$i \leftarrow 1, \ldots, \mathit{arity}(p)$}
      \State Create fresh variable $x_i$
      \State $\mathit{quantifiers} \leftarrow \mathit{quantifiers} \;.\;
      \forall x_i$
    \EndFor

    \State $\mathit{guard} \leftarrow \mathit{false}$
    \For{$(f, p(\bar s)) \in \mathit{requiredChecks}$}
    \Comment{Checks with symbol $p$}
    \State $\mathit{guard} \leftarrow \mathit{guard} \vee
                  (f \wedge \bar s = \langle x_1, \ldots, x_n \rangle)$
    \EndFor
    
    \State $\mathit{matrix} \leftarrow \mathit{matrix} \wedge
             (\mathit{guard} \to
                \textsc{EncodeBodies}(\mathit{nextClauses},
                                   \langle x_1, \ldots, x_n \rangle))$

    \EndWhile

    \Return{$\mathit{quantifiers} \;.\; \mathit{matrix}$}
    \EndFunction

\medskip
    \Function{EncodeBodies}{$\mathit{clauses}, \bar s$}
      \State $\mathit{result} \leftarrow \mathit{false}$
      \For{$C \wedge p_1(\bar t_1) \wedge \cdots \wedge p_n(\bar t_n) \to p(\bar t)
                \in \mathit{clauses}$}
      \State $\mathit{quantifiers} \leftarrow \mathit{quantifiers} \;.\;
      \exists\, \fv{C \wedge p_1(\bar t_1) \wedge \cdots \wedge p_n(\bar t_n) \to p(\bar t)}$
      \For{$i \leftarrow 1, \ldots, n$}
      \State Create fresh Boolean flag $f_i$
      \State $\mathit{quantifiers} \leftarrow \mathit{quantifiers} \;.\;
      \exists f_i$
      \State $\mathit{checksRequired} \leftarrow \mathit{checksRequired}
              \cup \{(f_i, p_i(\bar t_i))\}$
      \EndFor
      \State $\mathit{disjunct} \leftarrow
                  \bar t = \bar s \wedge C \wedge f_1 \wedge \cdots \wedge f_n$
      \State $\mathit{result} \leftarrow \mathit{result} \vee \mathit{disjunct}$
      \EndFor

      \Return{$\mathit{result}$}
    \EndFunction

  \end{algorithmic}

\bigskip
We illustrate the succinct encoding using an example. Consider the clauses
\begin{lstlisting}
(C1)  r(X,Y) $\expl$ Y = X + 1
(C2)  r(X,Y) $\expl$ Y = X + 2
(C3)  s(X,Z) $\expl$ r(X, Y) $\pand$ r(Y, Z)
(C4)  false  $\expl$ s(X, Z) $\pand$ X >= 0 $\pand$ Z <= 0
\end{lstlisting}
The formula resulting from the succinct encoding is:
\begin{lstlisting}
       $\exists$ x0, x1, f1. $\forall$ x3, x4. $\exists$ x5, x6, x7, f2, f3. $\forall$ x10, x11. $\exists$ x12, x13, x14, x15.
(C4)   (x1 >= 0 $\pand$ 0 >= x0 $\pand$ f1 $\pand$
         ((f1 $\pand$ x1 = x3 $\pand$ x0 = x4) $\to$
(C3)      (x7 = x3 $\pand$ x6 = x4 $\pand$ f2 $\pand$ f3)) $\pand$
         (((f2 $\pand$ x7 = x10 $\pand$ x5 = x11) $\por$
           (f3 $\pand$ x5 = x10 $\pand$ x6 = x11)) $\to$
(C1)      ((x13 = x10 $\pand$ x12 = x11 $\pand$ x12 = x13 + 1) $\por$
(C2)       (x15 = x10 $\pand$ x14 = x11 $\pand$ x14 = x15 + 2))))
\end{lstlisting}

\end{proof}


\section{Solvability of Recursion-free Horn Clauses over Presburger
  Arithmetic is co-NEXPTIME-Complete:
  Proof of Theorem~\ref{thm:co-NEXPTIME-complete}}
\label{app:NEXPTIME-hard}

  It has already been observed that solvability is in co-NEXPTIME, so
  we proceed to show hardness by direct reduction of
  exponential-time-bounded Turing machines (possibly
  non-deterministic, with binary tape) to recursion-free Horn clauses
  over quantifier-free PA. 
  A Turing machine~$M = (Q, \delta, q_0, F)$ is
  defined by
  \begin{itemize}
  \item a finite non-empty set~$Q$ of states,
  \item an initial state~$q_0 \in Q$,
  \item a final state~$f \in Q$,
  \item a transition relation~$\delta \subseteq ((Q \setminus \{f\})
    \times \{0,1\}) \times (Q \times \{0,1\} \times \{L, R\})$.
  \end{itemize}
  Wlog, we assume that $Q = \{0, 1, \ldots, f\} \subseteq \Z$ and $q_0 = 0$.

  We define a \relsym\ $\mathit{step}(q, l, r, q', l', r')$ to
  represent single execution steps of the machine. The parameters~$l,
  r, l', r'$ represent the tape, which is encoded as non-negative
  integers; the bits in the binary representation of the integers are
  the contents of the tape cells. $l$ is the tape left of the head,
  $r$ the tape right of the head. The least-significant bit of $r$ is
  the tape cell at the head position. $l', r'$ are the corresponding
  post-state variables after one execution step.

  A tuple~$(q, b, q', b', L) \in \delta$ (moving the tape to the left)
  is represented by a clause
  \begin{equation*}
    \mathit{step}(q,\; x,\; b + 2y,\; q',\; b' + 2x,\; y)
  \end{equation*}
  where $x, y$ are the implicitly universally quantified variables of
  the clause, and $q, b, q', b'$ concrete numeric
  constants. Similarly, a tuple~$(q, b, q', b', R) \in \delta$ is
  encoded as
  \begin{equation*}
    0 \leq x \leq 1
    \to
    \mathit{step}(q,\; x + 2y,\; b + 2z,\; q',\; y,\; x + 2b' + 4z)
  \end{equation*}
  To represent termination, we add a clause~$\mathit{step}(f, x, y, f,
  x, y)$, implying that the machine will stay in the final state~$f$
  forever.

  We then introduce $n$ further clauses to model an execution
  sequence of length $2^n$:
  \begin{align*}
    \mathit{step}(x, y, z, x', y', z') \wedge
    \mathit{step}(x', y', z', x'', y'', z'')
    &\to
    \mathit{step}^1(x, y, z, x'', y'', z'')
    \\
    \mathit{step}^1(x, y, z, x', y', z') \wedge
    \mathit{step}^1(x', y', z', x'', y'', z'')
    &\to
    \mathit{step}^2(x, y, z, x'', y'', z'')
    \\ & \cdots
    \\
    \mathit{step}\,^{n-1}(x, y, z, x', y', z') \wedge
    \mathit{step}\,^{n-1}(x', y', z', x'', y'', z'')
    &\to
    \mathit{step}^n(x, y, z, x'', y'', z'')
  \end{align*}

  The final clauses expresses that the Turing machine does not
  terminate within~$2^n$ steps, when started with the initial tape~$t$:
  ~~ $\mathit{step}^n(0, 0, t, f, x, y) \to \mathit{false}$.

  Clearly, the expansion~$\mathit{exp}(\ClauseSet)$ of the resulting
  set~$\ClauseSet$ of Horn clauses is unsatisfiable (i.e., $\ClauseSet$ can be
  solved) if and only if no execution of the Turing machine, starting
  with the initial tape~$t$, terminates  within~$2^n$ steps.

\section{Clauses Solvable Semantically but not Syntactically\label{example:nonlin}}

Consider the following clause set $\ClauseSet$:
\begin{lstlisting}
multA(X,Y,Z) $\expl$ X = 0 $\pand$ Z = 0
multA(X,Y,Z) $\expl$ multA(X1,Y,Z1) $\pand$ X1 = X - 1 $\pand$ Z = Z1 + Y
multB(X,Y,Z) $\expl$ X = 0 $\pand$ Z = 0
multB(X,Y,Z) $\expl$ multB(X1,Y,Z1) $\pand$ X1 = X - 1 $\pand$ Z = Z1 + Y
false $\expl$ multA(X,Y,Z1) $\pand$ multB(X,Y,Z2) $\pand$ Z1 $\neq$ Z2
\end{lstlisting}
The clauses define two version of a multiplication and assert that the result is functionally
determined by the first two arguments. Let $a,b \subseteq {\cal Z}^3$ denote
the interpretations of \lstinline{multA} and \lstinline{multB}, respectively,
in any solution that satisfies all Horn clauses.
We show that the only possibility is that $a = b = m$ where
$m = \{(x,y,z) \in {\cal Z}^3 \mid z = xy\}$ is the multiplication relation.
Indeed, by induction we can easily prove that $m \subseteq a$ and $m \subseteq b$,
using the first four clauses. To show the converse, suppose on the contrary, that
$(x,y,z) \in a$ where $z \neq xy$ (the case for $(x,y,z) \in b$ is symmetrical).
Because $(x,y,xy) \in b$ and $x \neq z$, the last clause does not hold, a contradiction.

Therefore, the clauses have a unique solution $a = b = m$,
but this solution is not definable in a Presburger
arithmetic (e.g. by semilinearity of the solution sets, or by
decidability of Presburger arithmetic vs undecidability of
its extension with multiplication). Therefore, the above
clauses give an example of clauses that are semantically but
not syntactically solvable in Presburger arithmetic.

Further such examples can be constructed by using Horn clauses to define other 
total computable functions that are not definable in Presburger arithmetic alone.

\section{Completeness of Horn Clause Verification: Proof of
  Lemma~\ref{lem:CEGAR-completeness}}
\label{app:CEGAR-completeness}

  ``$\Rightarrow$'':
  Define each \relsym\ $p$ as the disjunction $\bigvee_{(p, Q) \in S} \bigwedge Q$.
  Since $S$ is closed under the edge relation, this yields a solution
  for the set~$\ClauseSet$ of Horn clauses.

  ``$\Leftarrow$'': Suppose $\ClauseSet$ is syntactically solvable, with
  each \relsym\ $p$ being mapped to the constraint~$C_p$. We define
  the predicate abstraction~$\Pi(p) = \{C_p\}$, and construct the ARG
  with nodes~$S = \{(p, C_p)\}$, and the maximum edge
  relation~$E$, which is closed.

\end{document}